\documentclass[12pt]{article}

\usepackage{graphicx}
\usepackage{hyperref}
\usepackage{amsmath, amsfonts, amssymb, amsthm, float}

\usepackage[cmtip,all]{xy}

\usepackage[OT2,T1]{fontenc}

\DeclareSymbolFont{cyrletters}{OT2}{wncyr}{m}{n}

\DeclareMathSymbol{\Sha}{\mathalpha}{cyrletters}{"58}

\newtheorem{theorem}{Theorem}

\newtheorem{lemma}[theorem]{Lemma}

\newtheorem{prop}[theorem]{Proposition}

\theoremstyle{definition}

\newtheorem{definition}[theorem]{Definition}

\newtheorem{example}[theorem]{Example}

\theoremstyle{remark}

\newtheorem{remark}[theorem]{Remark}

\usepackage{url}
\usepackage{verbatim}
\usepackage{amsmath,amssymb,url,amsthm,verbatim,graphicx}
\usepackage{tikz}
\usetikzlibrary{positioning}
\usepackage{wrapfig}
\usepackage{booktabs,subfig}   % Improtant for professional tables
\usepackage{authblk}

\begin{document}
 
\newcommand{\m}{2r^3-\sqrt{3}r^{2.5}}
\newcommand{\mm}{r^3-\frac{\sqrt{3}}{2}r^{2.5}}
\newcommand{\nth}{n^\text{th}}
\newcommand{\rth}{r^\text{th}}
\newcommand{\ff}{\mathbb{F}}
\newcommand{\ffq}{\mathbb{F}_q}
\newcommand{\ffb}{\mathbb{F}_2}
\newcommand{\ffqm}{\mathbb{F}_{q^m}}
\newcommand{\fbm}{\mathbb{F}_{2^m}}
\newcommand{\kg}{G_2^{\otimes n}}
\newcommand{\kgg}{G^{\otimes n}}
\newcommand{\yy}{y_0^{N-1}}
\newcommand{\uu}{u_0^{N-1}}
\newcommand{\huu}{u_0^{i-1}}
\newcommand{\perr}{P_\text{err}}
\newcommand{\bb}{\mathcal{B}_{i,N}}
\newcommand{\aah}{\mathcal{A}_{i,N}}
\newcommand{\fii}{\mathcal{F}}
\newcommand{\sfrac}[2]{#1/#2}
\newcommand{\smm}{r^3-\sfrac{(\sqrt{3}}{2})r^{2.5}}
%%%  Text 
\newcommand{\sas}{\text{ Sas\u{o}glu }}
\newcommand{\valdut}{\text{V\u{a}ldut}}
\newcommand{\wt}{\text{wt}}
\newcommand{\argmax}{\operatorname{argmax}}
\newcommand{\divi}{\operatorname{Div}}
\newcommand{\Gali}{\operatorname{Gal}}
\newcommand{\resi}{\operatorname{res}}
\newcommand{\Auti}{\operatorname{Aut}}
\newcommand{\supi}{\operatorname{Supp}}
\newcommand{\phfi}{\operatorname{PHF}}
\newcommand{\prini}{\operatorname{Princ}}
\newcommand{\mati}{\operatorname{Mat}}
\newcommand{\erri}{\operatorname{err}}
\newcommand{\wti}{\operatorname{wt}}
\newcommand{\Tri}{\operatorname{Tr}}
\newcommand{\ddiv}{\operatorname{div}}
\newcommand{\PGU}{\operatorname{PGU}}
\newcommand{\wtt}{\operatorname{wt}}

\newcommand{\rmm}{\ensuremath{\text{RM}}}
\newcommand{\rmu}{\ensuremath{\text{RM}}}
\newcommand{\C}[2]{\ensuremath{{{#1}\choose{#2}}}}

%%%%%%%%%%%%%%%%%%%%%%%%%%%%%
%% Definition and Theorems %%
%%%%%%%%%%%%%%%%%%%%%%%%%%%%%

%\newtheorem{propo}{Proposition}

\newtheorem{lemm}{Lemma}  %%% was here

\newtheorem{notation}[theorem]{Notation}   %%% Was here 

\newtheorem{proposition}[theorem]{Proposition}  %%% Was here
%\newtheorem{definition}[theorem]{Definition}
%\newtheorem{example}[theorem]{Example}
%\newtheorem{exercise}[theorem]{Exercise}
%\newtheorem{axiom}{Axiom}
%\newtheorem{remark}{Remark}
%\newtheorem{claim}{Claim}
%\newtheorem{question}{Question}

%\numberwithin{equation}{chapter}

\def\X{{\mathcal{X}}}
\def\Y{{\mathcal{Y}}}
\def\ppp{{\mathbb{P}}}
\def\aaa{{\mathbb{A}}}
\def\fff{{\mathbb{F}}}
\def\qqq{\mathbb{Q}}
\def\rrr{\mathbb{R}}
\def\ccc{\mathbb{C}}
\def\zzz{\mathbb{Z}}
\def\nnn{\mathbb{N}}
\def\Sz{{\rm{Sz}}}
\def\Div{{\rm{Div}}}
\def\Aut{{\rm{Aut}}}

\def\supp{{\rm{Supp}}}
\def\Stab{{\rm{Stab}}}
\def\eval{{\rm{eval}}}
\def\wt{{\rm{wt}}}
\def\pf{\noindent {\bf Proof}:\ }
\newcommand{\SAGE}{{\sf SAGE}\xspace}
\newcommand{\sage}{\SAGE}

\newcommand{\fqq}{\mathbb{F}_q}
\newcommand{\fqqq}{\mathbb{F}_{q^4}}
\newcommand{\fffq}{\mathbb{F}_{q^{3}}}
\newcommand{\projective}[1]{\mathbb{P}^#1}
\newcommand{\divv}{\operatorname{div}}
\newcommand{\gf}{y^az^bh_1^ch_2^d}
\newcommand{\bcb}{\newline}

\newcommand{\fe}{\mathbb{F}_{q^4}}
\newcommand{\cod}[2]{C_{\mathcal{L}}(#1,#2)}
\newcommand{\valu}[2]{\nu_{#1}\left (#2 \right )}
\newcommand{\res}[2]{\operatorname{res}_{#1}(#2)}
\newcommand{\cf}{\fe(y)}
\newcommand{\place}[1]{\mathbb{P}_{#1}}
\newcommand{\pif}{P_\infty}

\newcommand{\pff}{\mathbb{P}_F}

\newcommand{\D}{\displaystyle}

\newcommand{\Gtwo}{G_{2}}
\newcommand{\Atwo}{A_{2}}
\newcommand{\Btwo}{B_{2}}

\newcommand{\GG}{^{2}G_{2}}
\newcommand{\AAA}{^{2}A_{2}}
\newcommand{\BB}{^{2}B_{2}}
\newcommand{\GF}[1]{\mathbb{F}_{#1}}
\newcommand{\gfq}{\mathbb{F}_{q}}
\newcommand{\fgq}{\mathbb{F}_{q}}
\newcommand{\fgqm}{\mathbb{F}_{q^m}}
\newcommand{\fg}[1]{\mathbb{F}_{#1}}
\newcommand{\fgqc}{\overline{\mathbb{F}_{q}}}

\newcommand{\JJr}{\mathcal{J}_{\text{R}}}

\newcommand{\NN}{\mathbb{N}}
\newcommand{\ZZ}{\mathbb{Z}}
\newcommand{\LL}{\mathcal{L}}

\newcommand{\xh}{X_{\text{H}}}
\newcommand{\xs}{X_{\text{S}}}
\newcommand{\xr}{X_{\text{R}}}
\newcommand{\fh}{F_{\text{H}}}
\newcommand{\fs}{F_{\text{S}}}
\newcommand{\fr}{F_{\text{R}}}
\newcommand{\gh}{g_\text{H}}
\newcommand{\gs}{g_\text{S}}
\newcommand{\gr}{g_\text{R}}

\newcommand{\reex}{\mathcal{X}}
\newcommand{\ree}{X_{\text{R}}}
\newcommand{\dimens}{\text{dim}}

\newcommand{\Wedge}{\land^2}
\newcommand{\OO}{\mathcal{O}}
\newcommand{\im}{\text{Im}}
\newcommand{\frrq}{\text{Fr}_q}

\newcommand{\pp}{\mathbb{P}}
\newcommand{\PP}{\mathbb{P}}
\newcommand{\ppfgqc}{\mathbb{P}^{13}(\fgqc)}
\newcommand{\ppfr}{\mathbb{P}_{\fr}}
\newcommand{\ppfx}{\mathbb{P}_{\fgq(x)}}

\newcommand{\Da}{\mathcal{D}}
\newcommand{\psii}{\psi_{\alpha \beta \gamma \delta}}
\newcommand{\Mat}[3]{\text{Mat}(#1,#2,#3)}

\newcommand{\gl}{\text{GL}(n,k)}

\newcommand{\fieldextension}{/}
\newcommand{\placeextension}{|}

\newcommand{\tps}{perfect hash families }
\newcommand{\tus}{$\epsilon$-almost strongly universal hash families}
\newcommand{\HH}{\mathcal{H}}
\newcommand{\PHF}{\text{PHF}(H;n,m,w)}
\newcommand{\PPHF}[4]{\text{PHF}\left ({#1;#2,#3,#4}\right )}
\newcommand{\RR}{\mathbb{R}}

\newcommand{\mat}[3]{\text{Mat}({#1},{#2},{#3})}
\newcommand{\xx}{\mathcal{X}}
\newcommand{\yyy}{\mathcal{Y}}

\newcommand \mapsfrom{\mathrel{\reflectbox{\ensuremath{\mapsto}}}}
%\title{Reseach Statment}
%\author{Abdulla Eid}
%\email{eid1@illinois.edu}
\pagenumbering{arabic}

\title{Using concatenated algebraic geometry codes in channel polarization}

%\author[1]{Abdulla Eid\thanks{eid1@illinois.edu}}
\author[1]{Abdulla Eid and Iwan Duursma\thanks{eid1@illinois.edu, duursma@math.illlinois.edu}}
%\affil[1]{Department of Mathematics, University of Illinois at Urbana-Champaign}
%\author{Abdulla Eid, Iwan Duursma}
%\email{duursma@illinois.edu,eid1@illinois.edu}
%\nodate
\maketitle

%\include{chapter00}
%\include{chapter0}

%\include{chapter00}

%************************************************
%******************* Abstract ******************
%************************************************

\begin{abstract}
Polar codes were introduced by Arikan \cite{A} in 2008 and are the first family of error-correcting codes achieving the symmetric capacity of an arbitrary binary-input discrete memoryless channel under low complexity encoding and using an efficient successive cancellation decoding strategy. Recently, non-binary polar codes have been studied, in which one can use different algebraic geometry codes to achieve better error decoding probability. In this paper, we study the performance of binary polar codes that are obtained from non-binary algebraic geometry codes using concatenation. For binary polar codes (i.e. binary kernels) of a given length $n$, we compare numerically the use of short algebraic geometry codes over large fields versus long algebraic geometry codes over small fields. We find that for each $n$ there is an optimal choice. For binary kernels of size up to $n \leq 1,800$ a concatenated Reed-Solomon code outperforms other choices.
For larger kernel sizes concatenated Hermitian codes or Suzuki codes will
do better.

\end{abstract}

%************************************************
%******************* Section 1 ******************
%************************************************

\section{Introduction to Channel Polarization}\label{sec4:1}
Polar codes were introduced by Arikan \cite{A} in 2008 and are the first family of error-correcting codes achieving the symmetric capacity of an arbitrary binary-input discrete memoryless channel under low complexity encoding and using an efficient successive cancellation decoding strategy. We introduce now the polar codes and the channel polarization phenomenon for a $q$-ary input symmetric discrete memoryless channel $W$ with input alphabet $\xx$ and output alphabet $\yyy$. We start first with some notations from \cite{USK},\cite{MT1'},\cite{MT1}.

\begin{notation}
Let $u_0^{N-1}$ be the vector $u=(u_0,\dots,u_{N-1}) \in \xx^{N}$ $(N \in \NN_{>0})$. For each $0 \leq i<j \leq N-1$, we denote by $u_i^j$ the subvector $(u_i,\dots,u_j)\in \xx^{j-i+1}$ of $u$. Moreover, if $\mathcal{F}:=\{f_0<\dots <f_t \}\subseteq \{0,1,\dots,N-1\}$ is a set of indices, then we denote by $u_{\mathcal{F}}$ the subvector $(u_{f_0},\dots,u_{f_t})\in \xx^{t+1}$ of $u$.
\end{notation}

Let $W:\xx \to \yyy$ be a $q$-ary input discrete memoryless channel with a uniform distribution on the input alphabet $\xx$. Let $\ell \geq 2$ be a positive integer and $g:\xx^\ell \to \xx^\ell$ be an $\fgq$-isomorphism linear map, which is called the \emph{kernel} map. Let $G \in \mati(\ell,\ell,\xx)$ be the matrix representing the map $g$ and $\kgg$ be the $n$th Kronecker product of $G$ of size $\ell^n \times \ell^n$.

\begin{definition}
For $i \geq 1$, the \emph{subchannel} $W^{(i)}:\xx \to \yyy^{\ell^n}\times \xx^{i}$ is defined as the channel on input $u_i$, output $(y_0^{\ell^n-1},u_0^{i-1})$, and probability distribution
\[
W^{(i)}(y_0^{\ell^n-1},u_0^{i-1}\,|\, u_i):= \frac{1}{q^{\ell^n}}\sum_{u_{i+1}^{\ell^n-1}\in \xx^{\ell^n-i+1}} W (y_0^{\ell^n-1}\,|\, u_0^{\ell^n-1}\kgg).
\]
\end{definition}
Let $\{B_i \}_{i\in \NN_{>0}}$ be a sequence of independent and identically distributed random variables defined over some probability space such that $B_i=k$ with probability $\sfrac{1}{\ell}$, for each $k=1,2,\dots,\ell$. Define the random process $\{W_n\,|\, n \in \NN\}$ recursively by
\begin{align*}
&W_0:=W,\\
&W_{n+1}:=W_n^{(B_{n+1})}.
\end{align*}

Recall that the symmetric capacity and the Bhattacharyya parameter of a $q$-ary input discrete memoryless channel $W:\xx \to \yyy$ is defined as
\[
I(W):= \frac{1}{q}\sum_{x \in \xx}\sum_{y \in \yyy}W(y \,|\, x) \log_{q} \frac{W(y \,|\, x)}{\frac{1}{q}\sum_{x' \in \xx} W(y \,|\, x')}. 
\]
and
\[
Z(W):=\frac{1}{q(q-1)} \sum_{\substack{x,x' \in \xx \\ x' \neq x}} \sum_{y \in \yyy} \sqrt{W(y \,|\, x)W(y \,|\, x')}.
\]
Set $I_n:=I(W_n)$ and $Z_n:=Z(W_n)$. Then,

\begin{lemma}\cite[Lemma 2]{STA} \cite[Lemma 9]{MT1'}
There exists a random variable $I_\infty$ such that $I_n \to I_\infty$ almost surely as $n \to \infty$.
\end{lemma}
Using the lemma above, the channel polarization occurs if the probability that $I_\infty \in \{0,1 \}$ is one. The term "polarization" refers to the fact that the subchannels \emph{polarize} to noiseless channels or pure-noisy channels \cite{A}.

\begin{definition}\label{def4:polarcode}
A \emph{polar code} is a code with kernel $G$ such that in the construction above, $Pr(I_\infty \in \{0,1 \})=1$ and
\begin{equation}\label{eq4:*}
I_\infty=\begin{cases}
1, \quad & \text{ w.p. } I(W),\\
0, \quad & \text{ w.p. } 1-I(W).
\end{cases}
\end{equation}
\end{definition}
From the definition above, we see that a repeated application of the matrix $G$ polarizes the underlying channel, i.e., the resulting subchannels $W^{(i)}$ ($i \in \{1,\dots,\ell^n \}$) tend toward either noiseless channels or pure-noisy channels. Moreover, Condition \eqref{eq4:*} guarantees that the fraction of the noiseless channels to all channels approaches $I(W)$, i.e., polar codes are capacity achieving. That suggests using the noiseless channels for transmitting the information symbols while transmitting no information over the pure-noisy channels, which are called the frozen symbols \cite{A}.

The polar codes introduced by Arikan \cite{A} in 2008 use the $2 \times 2$ matrix
\[
G_2=\begin{pmatrix} 1 & 0\\ 1 & 1\\ \end{pmatrix}.
\]
Thereafter, Urbanke, Korada, and \sas generalized Arikan's construction for BSCs. They showed that any $\ell \times \ell$ matrix $G$, none of whose column permutations is an upper triangular matrix, polarizes the channel. Mori and Tanaka \cite{MT1'} generalized the idea further to $q$-ary input S-DMC.

Polar codes have been found to be useful for many applications. They can be used to construct a lossy and lossyless source channel that achieve the rate-distortion trade-off with low encoding and decoding complexity, i.e., they have an optimal performance in that setting \cite{UHK}. Polar codes can also be used for deterministic broadcast channels \cite{GAG}, to achieve the secrecy capacity of wiretap channels \cite{MahdV}, and in ubiquitous computing and sensor network applications \cite{Ma}.

We conclude this section by giving sufficient condition for an non-identity $\ell \times \ell$ matrix $G$ to polarize the $q$-ary input S-DMC $W$. %First we assume $q$ is a prime integer. Then, the following theorem is due to \cite{USK} for $q=2$ and \cite{MT1'} for $q$ an odd prime. Here we will assume that $G$ is not the identity matrix, since the identity matrix does not polarize the channel $W$.

\begin{theorem}\label{Thm4:1.4}\cite[Theorem 4]{USK},\cite[Theorem 13]{MT1'} ($q$ is a prime integer)
Given a $q$-ary input symmetric discrete memoryless channel $W$, any $\ell \times \ell$ matrix $G$ none of whose column permutations is an upper triangular polarizes the channel.
\end{theorem}

\begin{remark}
Let $G$ be an $\ell \times \ell$ matrix and $U$ be an $\ell \times \ell$ upper triangular matrix. Then, the channels $W^{(i)}$ have the same statistical properties under $G$ and $GU$, i.e., $G$ and $GU$ are equivalent in the sense of Definition 4 in \cite{MT5}. Moreover, column permutations also do not change the statistical properties of $W^{(i)}$. Therefore, using an $LUP$ decomposition we may assume that $G$ itself is a lower triangular matrix.
\end{remark}

\begin{theorem}\label{Thm4:1.5} \cite[Theorem 11]{MT5} ($q$ is a prime power)
Let $\fgq$ be a non-prime finite field of characteristic $p$. Given a $q$-ary input symmetric discrete memoryless channel $W$, then an $\ell \times \ell$ lower triangular matrix $G$ polarizes the channel $W$ if and only if $\fgq=\mathbb{F}_{p}(G)$, where $\mathbb{F}_{p}(G)$ is the field extension of $\mathbb{F}_{p}$ generated by the entries of $G$.
\end{theorem}

%************************************************
%******************* Section 2 ******************
%************************************************

\section{The Performance of Polar Codes and the Rate of Polarization}\label{sec4:2}

In this section we explain the asymptotic error probability and how it depends on the exponent of the kernel, a quantity that plays an important role in determining the performance of polar codes. 

\subsection{The Error Probability}\label{sec4:2.4}
The performance of the polar code is determined based on the asymptotic error probability. In Arikan's construction of a binary polar code \cite{AT} using the matrix $G_2$, the asymptotic error probability of the polar code using the successive cancellation decoding is
\begin{equation}\label{sec4:eqG2}
\perr= \begin{cases}
o\left (2^{ -N^\beta} \right ), &\quad \beta <\frac{1}{2},\\
\omega \left ( 2^{ -N^\beta}\right ), &\quad \beta \geq \frac{1}{2}.
\end{cases}
\end{equation}
The error probability $P_{\erri}$ above\footnote{The parameter $\beta$ is an arbitrary real number, if $\beta<\sfrac{1}{2}$, then there exists a polar code that satisfies $\perr=o\left (2^{ -N^\beta}\right )$} is independent of the rate $R$ of the polar code. Mori and Tanaka \cite{MT3} extended the formula above to one that is rate dependent. Moreover, the threshold $\sfrac{1}{2}$ depends only on the matrix $G_2$ and not on the underlying channel $W$.

Now for any $\ell \times \ell$ matrix $G$ that polarizes the $q$-ary input S-DMC $W$, the asymptotic error probability of the polar code is given by
\begin{equation}\label{sec4:eqEG}
\perr= \begin{cases}
o\left (2^{ -N^\beta}\right ), &\quad \beta <E(G),\\
\omega \left (2^{ -N^\beta}\right ), &\quad \beta \geq E(G)
\end{cases}
\end{equation}
for some well-defined constant $E(G)\in [0,1)$ depending only on the matrix $G$ and not on the underlying channel $W$. The quantity $E(G)$ is called the \emph{exponent} of the matrix $G$. It measures the performance of the polar code under successive cancellation decoding (\cite[Section IV]{USK} for $q=2$ and \cite[Theorem 9]{MT1'} for any $q>2$).

\begin{remark}
As in \cite[Theorem 19]{MT1'}, Equation~\eqref{sec4:eqEG} can be read as follows. For any kernel which polarizes the undelying channel it holds
\[
\lim_{n \to \infty} Pr(Z_n < 2^{\ell^{-n\beta}}) = I(W)
\]
for $\beta < E(G)$ and 
\[
\lim_{n \to \infty} Pr(Z_n < 2^{\ell^{-n\beta}}) = 0
\]
for $\beta > E(G)$.

\end{remark}

\subsection{The Formula of the Exponent}\label{sec4:2.5}
As mentioned in the previous subsection, the exponent of the matrix $G$ plays an important role in determining the performance of the polar codes. In \cite{USK},\cite{MT1}, the authors gave an algebraic description of the exponent in terms of the partial distances of the matrix $G$. 

\begin{definition}
 Given an $\ell \times \ell$ matrix $G=(g_1,\dots,g_\ell)^{T}\in \mat{\ell}{\ell}{\fgq}$, the \emph{partial distances} $D_i$ $(i=1,2,\dots,\ell)$ are defined by
\begin{align*}
D_\ell&:=\wtt(g_\ell)=d(g_\ell,0),\\
D_i&:=d(g_i,\langle g_{i+1},\dots,g_{\ell} \rangle ),
\end{align*}
where $d(g_i,\langle g_{i+1},\dots,g_{\ell} \rangle )$ is the distance from the codeword $g_i$ to the code generated by the codewords $g_{i+1},\dots,g_{\ell}$. The partial distances $\{D_i \}_{i=1}^{\ell}$ of a matrix $G$ will be called the \emph{profile} of $G$.
\end{definition}

\begin{theorem}\label{sec4:Thm2.5}\cite[Theorem 14]{USK},\cite[Theorem 1]{MT1}
 Let $W$ be a $q$-ary input symmetric discrete memoryless channel and $G$ be an $\ell \times \ell$ matrix that polarizes the channel $W$ with profile $\{D_i \}_{i=1}^{\ell}$. Then, the exponent of $G$ is given by
\begin{equation}\label{sec4:2.*}
 E(G)=\frac{1}{\ell}\sum_{i=1}^{\ell}\log_\ell D_i. 
\end{equation}
\end{theorem}

For $q$-ary codes, Mori and Tanaka \cite{MT1} suggested using algebraic geometry codes in order to get larger exponents. The motivation behind this idea is the Reed-Solomon code which gives large exponents. In particular, algebraic geometry codes have in general large minimum distance and often they have a nested structure similar to the Reed-Solomon code which makes them suitable for channel polarization.

\subsection{Concatenation}\label{sec4:4.1}

In this section we introduce the \emph{concatenation} of codes which is illustrated in the following theorem.

\begin{theorem}\label{Thm4:4.1}\cite[Theorem 6.3.1]{XL}
Let $C_1$ be a $(N,K,D)$-linear code over $\fg{q^m}$ and $C_2$ be a $(n,m,d)$-linear code over $\fg{q}$. Then, there exists a\\ $(nN,mK,dD)$-linear code $C$ over $\fg{q}$.
\end{theorem}

The code $C_1$ in Theorem \ref{Thm4:4.1} is called the \emph{outer code}, the code $C_2$ is called the \emph{inner code}, and the code $C$ is called the \emph{concatenated code}. In the following, we will use the descent code $(\fg{q})^m$ which is a $(m,m,1)$-linear code as an inner code and all field extensions $\fg{q}$ are of characteristic 2. Therefore, given a $(N,K,D)$-linear code over $\fg{2^m}$, Theorem \ref{Thm4:4.1} yields a binary $(mN,mK,D)$-linear code.

Let $G \in \mat{L}{L}{\fg{q}}$ be a generating matrix for a code over $\fg{q}$ and let $(D_1,\dots,D_L)$ be its profile. Recall that the exponent of $G$ is given by
\[
 E(G):=\frac{1}{L} \sum_{i=1}^{L} \log_L D_i = \frac{1}{L \log_2 L} \sum_{i=1}^{L} \log_2 D_i.
\]
Applying the inner code $(m,m,1)$ to the matrix $G$ will replace each symbol in $G$ with $m$ binary symbols and each row will be redundant $m$ times. Therefore, the new concatenated matrix, denoted by $G_2$, will be of size $mL \times mL$, and will have at least the following profile
\[
 \underbrace{D_1,\dots, D_1}_{m\text{-times}},\underbrace{D_2,\dots, D_2}_{m\text{-times}},\dots,\underbrace{D_L,\dots, D_L}_{m\text{-times}}.
\]
Then, the exponent of the binary matrix $G_2$ satisfies the inequality
\begin{align*}
E(G_2):&=\frac{1}{mL\log_2(mL)}\cdot m\cdot \sum_{i=1}^{L}\log_2 D_i\\
       &\geq \frac{1}{L\log_2(mL)}\cdot L \cdot \log_2 L\cdot  E(G)\\
       &\geq \frac{\log_2L}{\log_2(mL)}\cdot E(G).
\end{align*}

\begin{comment}
\begin{remark}
 If we use different inner codes of dimension $m$, e.g., $(m+1,m,2)$ (see Example \ref{sec0:InnerCodeExample}), then the exponent of the concatenated binary matrix $G_2$ satisfies the inequality
\begin{align*}
E(G_2)&=\frac{1}{(m+1)L\log_2((m+1)L)}((m+1)\left (\log_2 D_1+\dots+\log_2 D_L \right )\\
       &\geq \frac{1}{L\log_2((m+1)L)}\cdot L \log_2 L \cdot E(G)\\
       &\geq \frac{\log_2 L}{\log_2 ((m+1)L)}\cdot E(G).
\end{align*}
Similarly, if we use the first-order Reed-Muller code $(2^m,m,2^{m-1})$ (see Example \ref{sec0:InnerCodeExample}) as an inner code, we get
\[
E(G_2) = \frac{\log_2 L}{\log_2 (2^mL)} E(G).
\]
In both cases, the descent inner code gives a better lower bound.
\end{remark}

\end{comment}

\subsection{Our Results}\label{sec4:2.6}

Mori and Tanaka \cite{MT1} evaluated the performance of polar codes using kernels constructed from the generating matrices of the Reed-Solomon and Hermitian codes over a $q$-ary field. They have found that numerically Hermitian codes give larger exponents than Reed-Solomon codes. That suggests using different algebraic geometry codes as they have large minimum distance and often have the same nested structure as Reed-Solomon codes. We continue in this direction by using algebraic geometry codes to study the behavior of the exponent. In Section \ref{sec4:3} we show for a subclass of algebraic geometry codes, that $E(G_L)\to 1$ as the number of affine rational points $L \to \infty$. 

In Section \ref{sec4:4} we will apply the concatenation of codes to construct binary codes from $q$-ary codes. As the algebraic geometry codes are defined over different field extensions of $\fg{2}$. This helps us to study the performance of different algebraic geometry codes defined over a common field which is the binary field $\fg{2}$. We will study numerically whether to use larger field extensions or curves with many rational points to get a larger exponent. In other words, it is the study of how to approach $\infty$ in the most efficient way using either concatenation or geometry. In the first case (concatenation) we would need codes defined over large fields and in the second case (geometry) we would need curves with many rational points defined over small fields. 

In Section \ref{sec4:5} we compare numerically how the Reed-Solomon, Hermitian, and Suzuki codes behave in the settings above for a given binary block size. It turns out that each code will give the maximum exponent for some range of the binary block size, and that more geometry is preferable as the block size increases. We also compare their performance as error-correcting codes.

%************************************************
%******************* Section 3 ******************
%************************************************

\section{The Algebraic Geometry Codes in Channel Polarization}\label{sec4:3}
In this section we recall first the construction of the algebraic geometry codes. We give three examples of algebraic geometry codes, the Reed-Solomon code, the Hermitian code, and the Suzuki code. Moreover, we will show in this section that for a subclass of algebraic geometry codes with block length $L$ and generating matrix $G_L$, we have $E(G_L) \to 1$ as $L \to \infty$.

\subsection{Algebraic Geometry Codes as Kernels for Channel Polarization}\label{sec4:3.1}

Let $X \fieldextension \fgq$ be a curve (smooth, irreducible, and projective) of genus $g$ over $\fgq$ with global function field $F:=\fgq(X)$. Let $X(\fgq) $ be the set of all $\fgq$-rational points on $X$ with cardinality $N(F)$. Let $1 \leq n < N(F)$ and choose $n$ distinct rational places $P_1,\dots,P_n \in X(\fgq)$. Set $D:=P_1+ \dots +P_n$ and let $G':=\sum_{i=1}^{s}n_iQ_i- \sum_{j=1}^{t}m_jQ'_j \in \divi(X)$ ($n_i,m_j \in \NN$) be a divisor of $X$ such that $\supi(D) \cap \supi(G')=\phi$. Define the Riemann-Roch space $\mathcal{L}(G')$ to be the $\fgq$-vector space of all rational functions $f\in F$ which are having only poles at $Q_i$ of order less than or equal to $n_i$ and zeros at $Q'_i$ of order greater than or equal to $m_j$, i.e.,
\[
\mathcal{L}(G'):=\{ f \in F \,|\, (f)+G' \geq 0 \} \cup \{0\}
\]

 Then, the algebraic geometry code is defined as
\[
 C_\mathcal{L}(D,G'):=\{ (f(P_1),\dots, f(P_n))\in \ffq^n \,|\, f \in \mathcal{L}(G') \}.
\]
The code $ C_\mathcal{L}(D,G')$ is a $(n,k,d)$-linear code over $\fgq$, where $k:=\ell(G')-\ell(G'-D)$ and $d \geq d^*:=n-\deg(G')$. Moreover, assume $n > \deg(G')$ and let $\{f_1, \dots, f_{k} \}$ be a basis for $\mathcal{L}(G')$ over $\fgq$. Then, $C_\mathcal{L}(D,G')$ has the following generator matrix
\[
 G:=\left( f_{i}(P_j) \right )_{\substack{i=1,\dots,k\\ j=1,\dots,n}}. %%%% Fix the notation here
\]
\begin{example} (Reed-Solomon Code)\label{sec4:ex3.1}
 Let $F:=\fgq(x)$ be the rational function field that corresponds to the projective line $\mathbb{P}^{1}(\fgq)$. Let $P_1,\dots,P_q\in \mathbb{P}_F$ be the affine rational places of $F$. Set $G':=dP_\infty$ ($d < q$). The set $\{1,x,x^2,\dots,x^{d} \}$ is a basis for $\LL(dP_\infty)$ over $\fgq$. Then, the Reed-Solomon code of length $q$ is the algebraic geometry code $C_\mathcal{L}(P_1 + \dots + P_q,dP_\infty)$.
\end{example}

\begin{example} (Hermitian Code)\label{sec4:ex3.2}
Let $\fh:=\fgq(x,y)$ be the Hermitian function field of genus $\gh=\sfrac{q_0(q_0-1)}{2}$ over $\fgq$ $(q:=q_0^2$ and $q_0$ is a prime power) defined by the equation $y^{q_0}+y=x^{q_0+1}$. Let $P_1,\dots,P_{q_0^3} \in \mathbb{P}_{\fh}$ be the affine $\fgq$-rational places of $\fh$. Set $G':=dP_\infty$ ($d<q_0^3$). The set $\{x^iy^j \mid i \geq 0, iq_0+j(q_0+1)\leq d \}$ is a generating set for $\LL(dP_\infty)$ over $\fgq$ \cite[Lemma 6.4.4]{Sti}. Then, the Hermitian code of length $q_0^3$ is the algebraic geometry code $C_\mathcal{L}(P_1 + \dots + P_{q_0^3},dP_\infty)$. 
\end{example}

\begin{example} (Suzuki Code)\label{sec4:ex3.3}
Let $\fs:=\fgq(x,y)$ be the Suzuki function field of genus $\gs=q_0(q-1)$ over $\fgq$ $(q:=2q_0^2$, $q_0=2^m$, and $m \in \NN$) defined by the equation $y^{q}-y=x^{q_0}(x^q-x)$. Let $P_1,\dots,P_{q^2} \in \mathbb{P}_{\fs}$ be the affine $\fgq$-rational places of $\fs$. Set $G':=dP_\infty$ ($d<q^2$). The set $\{x^ay^bz^{c}w^{d'} \mid a,b,c,d'\geq 0, aq+b(q+q_0)+c(q+2q_0)+d'(q+2q_0+1)\leq d \}$ is a generating set for $\LL(dP_\infty)$ over $\fgq$ \cite{HS}, where $z:=x^{2q_0+1}-y^{2q_0}$ and $w:=xy^{2q_0}-z^{2q_0}$. Then, the Suzuki code of length $q^2$ is the algebraic geometry code $C_\mathcal{L}(P_1 + \dots + P_{q^2},dP_\infty)$.
\end{example}

\subsection{The kernel $E(G)$ for Algebraic Geometry Codes}\label{sec4:3.2}
In this section we use Stirling's formula to show that $E(G_L) \to 1$ as $L \to \infty$, where $G_L$ is the generating matrix of an algebraic geometry code of block length $L$. First we state the Oesterl\'{e} bound (see \cite[Theorem 8, page 130]{Hurt}) which gives a lower bound to the genus $g$ of a curve $X \fieldextension \fgq$ with $L$ affine rational points. Let $\ell$ be the unique integer such that $\sqrt{q}^\ell < L \leq \sqrt{q}^{\ell +1}$, i.e.,
\begin{equation}\label{eq:oell}
\ell=\left \lceil \frac{\log_2 L}{\log_2 \sqrt{q}}-1 \right \rceil.
\end{equation}
We find 
\begin{equation}\label{eq:ou}
u:=\frac{\sqrt{q}^{\ell+1}-L}{L\sqrt{q}-\sqrt{q}^\ell}\in [0,1). 
\end{equation}
Next, we find $\theta \in \left [\sfrac{\pi}{(\ell+1)},\sfrac{\pi}{\ell} \right )$ such that 
\begin{equation}\label{eq:otheta}
 \cos \left (\frac{\ell+1}{2}\theta \right)+u \cos \left(\frac{\ell-1}{2}\theta\right)=0.
\end{equation}
Then, the Oesterl\'e bound is the lower bound
\begin{equation}\label{eq:og}
 g \geq \frac{(L-1)\sqrt{q} \cos \theta +q -L}{q+1-2\sqrt{q} \cos \theta}.
\end{equation}
Note that for some small $L$, this bound can be achieved by maximal curves over $\fg{q^2}$ or for large $L$ by towers of function fields. We will study the case where $g := \sfrac{\left( (L-1)\sqrt{q} \cos \theta +q -L \right )}{\left ( q+1-2\sqrt{q} \cos \theta\right )}$.

\begin{remark}
 If $L$ is small relative to the field size $q$ (e.g., $L < \sqrt{q})$, then the right hand side of the Inequality \eqref{eq:og} is negative. In that case we take $g=0$ as the code can be achieved by Reed-Solomon codes as in Example \ref{sec4:ex3.1}.

%%%% This is commented as it is much details that we don't need

\begin{comment} 
 
  \item Solving Equation \eqref{eq:otheta} is equivalent to solving the polynomial equation
\[
 T_{\ell+1}(x)+uT_{\ell-1}(x)=0, \quad x^2\in \left[\frac{\cos\left(\frac{\pi}{\ell+1}\right)}{2}+1,\frac{\cos\left(\frac{\pi}{\ell}\right)}{2}+1\right),
\]
%\[
% T_{\ell+1}(x)+uT_{\ell-1}(x)=0, \quad x^2\in \left[\sfrac{\cos\left(\frac{\pi}{\ell+1}\right)}{2}+1,\sfrac{\cos\left(\frac{\pi}{\ell}\right)}{2}+1\right),
%\]
where 
\[
 T_n(x)= \sum_{k=0}^{\left \lfloor \frac{n}{2} \right \rfloor} (-1)^k\binom{\left \lfloor \frac{n}{2}\right \rfloor}{2k} x^{n-2k}(1-x^2)^k 
\]
is the Chebyshev polynomial which is hard to solve for $n \geq 5$.

\end{comment}

\end{remark}

Now we prove the result of this section. Here we will study only the subclass of algebraic geometry codes over $\fgq$ with a matrix $G_L \in \mat{L}{L}{\fgq}$ that has profile\footnote{The least possible profile for a nested structure AG code using the point at infinity}
\begin{equation}\label{eq:AGprofile}
 D_i := \begin{cases}
          L-g+1-i, \qquad &i=1,2,\dots,L-g,\\
          1, \qquad & i=L-g+1,\dots,L.
        \end{cases}
\end{equation}
Therefore,
\begin{equation}\label{eq:AGexponent}
 E(G_L)=\frac{\log_q ((L-g)!)}{L\log_q L}.
\end{equation}

\begin{prop}\label{prop4:3.1}
 For the class of algebraic geometry codes over a fixed field $\fgq$ of length $L$ and matrix $G_L$ with profile $D_i=\max(L-g+1-i,1)$ $(i=1,\dots,L))$, we have
\[
E(G_L) \to 1 \qquad \text{ as }\qquad L \to \infty. 
\]
\end{prop}

\begin{proof}
 Recall Stirling's formula \cite{Stirling}
\[
 \lim_{n \to \infty} \frac{n!}{\sqrt{2\pi n}\left( \frac{n}{e} \right )^n} = 1
\]
which is equivalent to 
\[
 n! \simeq \sqrt{2 \pi n}\left( \frac{n}{e} \right )^n.
\]
Therefore, we get the estimate
\[
 \log (n!) = n \log n -n + O(\log n), \quad \text{ where } O(\log n)\simeq \frac{1}{2}\log (2\pi n).
\]
Now we have
\begin{align*}
E(G_L)&=\frac{\log_q ((L-g)!}{L\log_q L}\\
      &=\frac{1}{L\log L}\left [(L-g)\log(L-g)-(L-g)+O(\log(L-g))\right ]\\
      &= \frac{L-g}{L} \cdot \frac{\log(L-g)}{\log(L)} - \frac{L-g}{L\log(L)} + \frac{1}{2} \cdot \frac{\log(2\pi (L-g))}{L\log(L)}.
\end{align*}
Recall that 
\[
 g=\frac{(L-1)\sqrt{q}\cos \theta+q-L}{q+1-2\sqrt{q}\cos \theta}= \frac{L(\sqrt{q}\cos \theta -1)-\sqrt{q}\cos \theta+q}{q+1-2\sqrt{q}\cos \theta}.
\]
Then,
\begin{align*}
 L-g &= L\left [1- \frac{(\sqrt{q}\cos \theta-1)-\frac{\sqrt{q}\cos \theta}{L}+\frac{q}{L}}{q+1-2\sqrt{q}\cos \theta}\right]\\
    &\simeq L \left[1- \frac{\sqrt{q}\cos \theta -1}{q+1-2\sqrt{q}\cos \theta}\right] \simeq L\left[\frac{q+2-3\sqrt{q}\cos \theta}{q+1-2\sqrt{q}\cos \theta}\right]\\
    &\simeq a(q)L,
\end{align*}
for some constant $a(q)$ close to $1$ for large $q$. Therefore, we have that 
\begin{align*}
& \frac{L-g}{L \log (L)}\simeq \frac{a(q)}{\log (L)}\to 0 &&\text{ as } L\to \infty, \\
& \frac{\log(2\pi(L-g))}{L\log (L)}\simeq \frac{\log (2\pi a(q)L)}{L\log (L)}\to 0 &&\text{ as } L \to \infty, \\
& \frac{L-g}{L}\cdot \frac{\log(L-g)}{\log(L)}\simeq \frac{a(q)L}{L}\cdot \frac{\log(a(q)L)}{\log(L)} \to a(q)\to 1 &&\text{ as } L \to \infty.
\end{align*}
Therefore,
\begin{align*}
 E(G_L)&\simeq \underbrace{\frac{L-g}{L}\frac{\log(L-g)}{\log(L)}}_{\to 1}-\underbrace{\frac{L-g}{L\log(L)}}_{\to 0}+\underbrace{\frac{1}{2}\frac{\log(2\pi (L-g))}{L\log(L)}}_{\to 0}.\\
       &\to 1.
\end{align*}
as $L \to \infty$
\end{proof}

%\begin{remark}
% In general, for any algebraic geometry code with kernel $G$, the exponent $E(G)$ has a profile that is greater than \eqref{eq:AGprofile}. 
%\end{remark}

%************************************************
%******************* Section 4 ******************
%************************************************

\section{Fixing the Parameters}\label{sec4:4}

We have seen in the previous section that asymptotically $E(G_L)\to 1$ as $L \to \infty$. In this section we would like to study numerically\footnote{All the computations in this section are performed by Mathematica \cite{Mathematica} } the asymptotic behavior of different algebraic geometry codes. In order to compare their performance, we need these codes to be defined over a common field which is the binary field $\fg{2}$. For that, we use the concept of concatenation that was first introduced by Forney \cite{Forney11} as a technique to obtain new codes over the binary field from codes over a field extension of $\fg{2}$, see Section~\ref{sec4:4.1}.

We keep the same assumption about $G$ as in Section \ref{sec4:3}, i.e., $G$ is the generating matrix of an algebraic geometry code constructed using an algebraic curve of genus $g$ with $L$ affine rational points and profile $D_i=\max(L-g+1-i,1)$ $(i=1,\dots,L))$. As in Section \ref{sec4:4.1}, the exponent of $G$ and $G_2$ are related by the inequality
\begin{equation}\label{eq4:4.2}
E(G_2) \geq \frac{\log_2L}{\log_2(mL)}\cdot E(G).
\end{equation}
Therefore, using Proposition \ref{prop4:3.1}, $E(G_{L,2}) \to 1$ as $L \to \infty$.

We would like to study the asymptotic behavior of $E(G_2)$ by studying the lower bound in Inequality \eqref{eq4:4.2}. Set 
\begin{align*}
 E_2(G):&= \frac{\log_2 L}{\log_2(mL)}\cdot E(G)= \frac{\log_2(L)}{\log_2 (mL)}\cdot \frac{1}{L \log_2 (L)}\cdot \log_2((L-g)!)\\
        &= \frac{1}{L\log_2(mL)} \log_2 ((L-g)!)
\end{align*}
which is a function in $m$ (field size), $L$ (number of affine rational points), and $\theta$ (algebraic curve) \eqref{eq:og}. We denote the binary block length after concatenation by $n$, i.e., $n:=mL$ and we will regard $E_2(G)$ as a function on $n$, $m$, and $\theta$.

\begin{comment}
In this section we study the behavior of $E_2(G)$ as a single-variable function on either $n$, $m$, or $\theta$ by fixing the other parameters. First we fix $n$ and we regard $E_2(G)$ as a function on $m$. It turns out that $E_2(G)$ has a local maximum. Second we fix $m$. Then, as in Section \ref{sec4:3}, $E_2(G)$ is an increasing function on $n$ with limit 1 as $n \to \infty$. Finally if we regard $E_2(G)$ as a function on $\theta$, for fixed $n$ and $m$, then studying the behavior of $E_2(G)$ is the same as solving a polynomial-logarithmic equation of high degree, which is - in general - hard to be solved.
\end{comment}

\subsection{Fixing the Binary Block Size $n$}\label{sec4:4.2.1}
Let $n$ be a fixed binary block size. Then, $L=L(m)=\sfrac{n}{m}$ and $\theta=\theta(m)$ can be found using \eqref{eq:otheta}. This means that $E_2(G)$ can be written in term of $m$. Numerically using Mathematica, enumerating over all positive integers $m$ shows that $E_2(G)$ has a local maximum, e.g., see Figure \ref{fixn} and Table \ref{table3220} for binary block size $n=3\cdot 2^{20}$.

\begin{figure}[htb!]
  \centering
  \includegraphics[height=30cm,width=10.0cm,keepaspectratio]{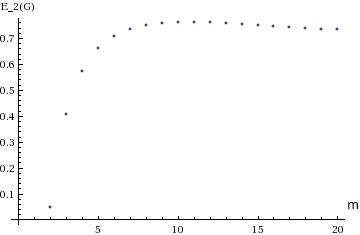}
  \caption[fixn]
  {Fixing $n=3\cdot 2^{20}$ and enumerating over all $m$, we have a local maximum at $m$=12 which correspond to the Hermitian curve.}\label{fixn}
\end{figure}

%Table~\ref{table3220} gives the numerical values for the figure~\ref{fixn}.
\begin{table}[htb!]
  \begin{center}
     \begin{tabular}[htb!]{ | c || c | c | c | c | }
      \hline
      $m$   & $q$ & $L$ & $g$ & $E_2(G)$ \\ \hline  
      2   &   2  &      1572864 &   $1.46820*10^6$  &  0.046959\\ \hline
      3  &     8   &     1048576  &  524647.    &     0.406401\\ \hline	
      4   &   16    &   786432   &  233948.      &   0.573893\\ \hline	
      6   &   64    &   524288   &	   62517.2 &        0.708937\\ \hline	
      8   &   256   &   393216   &  19901.6   &      0.750686\\ \hline	
      12  &   $2^{12}$  &   262144   &  2016.00   &      0.760667\\ \hline	
      16  &   $2^{16}$  &   196608   &  256.000   &      0.746789\\ \hline	
      24  &   $2^{24}$  &   131072   &      0     &      0.720751\\ \hline	
      32  &   $2^{32}$  &   98304    &      0     &      0.701524\\ \hline	
    \end{tabular}
    \caption{The values of $E_2(G)$ for $n=3\cdot 2^{20}$, where the maximum is at $m=12$ which corresponds to the Hermitian curve.}\label{table3220}
  \end{center}
\end{table}
In Section \ref{sec4:5}, numerically as $n$ gets larger, this local maximum corresponds first to the Reed-Solomon code for a small range of $n$, then it corresponds to the Hermitian code and as $n$ gets larger, the local maximum corresponds to the Suzuki code (see Figure \ref{fig:RSHS2}).

%%%%% Put here the figure of RS | H | S
\begin{comment}
\begin{figure}[htb!]
\centering

\begin{tikzpicture}
  [scale=1,auto=left, minimum size=5em]

   %%%%%%%%%%%%%% The following for the two Rectangles 
 
 \draw [<-] (1,4) --  (1,1); 
 \draw [->] (1,1) -- (8,1);

  \draw  (3.5,2) --  (3.5,3); 
  \draw  (5.5,2) --  (5.5,3);
  
 \node [align=center] (1) at (2.1,2.5) {Reed-Solomon};
 \node [] (2) at (4.5,2.5) {Hermitian};
 \node [align=center] (3) at (6.3,2.5) {Suzuki};

  \node [align=center] (1) at (3.5,0.7) {$1,800$}; 
   \node [align=center] (1) at (5.5,0.7) {$399,212$};

 \node [align=center] (1) at (8.3,1) {$n$}; 
 \node [align=center] (1) at (0.1,3.75) { Maximum\\ $E_2(G)$}; 
  
\end{tikzpicture}
\caption{ The maximum exponent corresponds to the three curves above.}\label{fig:RSHS}
\end{figure}
\end{comment}

\subsection{Fixing the Algebraic Curve $\theta$}\label{sec4:4.2.2}
Let $\theta$ be fixed, i.e., the underlying algebraic curve is fixed. Then, $L=L(m)$ can be found in terms of $m$ as follows. Find $\ell:=\ell(\theta)=\left \lfloor \sfrac{\pi}{\theta} \right \rfloor$, using $\eqref{eq:otheta}$ we have
\[
 u=-\frac{\cos \left ( \left (\frac{\ell+1}{2} \right )\theta \right )}{\cos \left ( \left (\frac{\ell-1}{2} \right )\theta \right )}.
\]
Then, using \eqref{eq:ou}, we get 
\[
 L:=L(m)=\frac{\sqrt{q}^{\ell+1}-u\sqrt{q}^{\ell}}{u\sqrt{q}+1}=\frac{\sqrt{2^{m}}^{\ell+1}-u\sqrt{2^{m}}^{\ell}}{u\sqrt{2^{m}}+1}.
\]
Therefore, $E_2(G)$ is a function of $m$. In this case it turns out that $E_2(G)$ is an increasing function (see Figure \ref{fixthetafreem}). This means that \emph{more concatenation gives larger exponents}.
\begin{figure}[htb!]
  \centering
  \includegraphics[height=30cm,width=10.0cm,keepaspectratio]{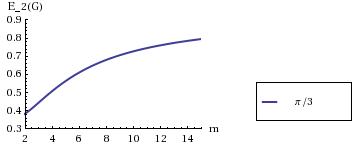}
  
  \caption[fixthetafreem]
  {Fixing $\theta=\sfrac{\pi}{3}$ and enumerate over all $m$.}\label{fixthetafreem}
\end{figure}

\subsection{Fixing the Field Extension $\fg{2^m}$}\label{sec4:4.2.3}
Let $m$ be fixed. Using \eqref{eq:otheta}, $\theta=\theta(L)$ and so $E_2(G)$ can be regarded as a function of $L$. Numerically $E_2(G)$ is an increasing function in $L$ with limit that tends to $1$ as $L \to \infty$ (see Figure \ref{fixmfreel}). This means again that \emph{more geometry gives larger exponents}.

\begin{figure}[htb!]
  \centering
  \includegraphics[height=30cm,width=10.0cm,keepaspectratio]{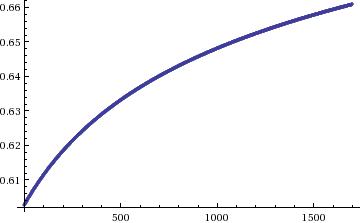}
  
  \caption[fixmfreel]
  {The exponent $E_2(G)$ if we fix $m=8$ and enumerate on $L$.}\label{fixmfreel}
\end{figure}

\begin{remark}
 From the two analyses above, we see that more concatenation gives larger exponents and similarly, more geometry also gives larger exponents. It is the question about the tradeoff between concatenation and geometry in how to efficiently approach the limit, i.e., $E(G)\to 1$. In Section \ref{sec4:5}, we will see that numerically for the three codes (Reed-Solomon, Hermitian, and Suzuki codes) more geometry is preferable as the block size increases.
\end{remark}

\begin{comment}

\begin{remark}
 If we would like to study $E_2(G)$ as a function of $\theta$, then we need to find $L$ in terms of $m$ or $m$ in terms of $L$, i.e., we need to solve the system of equations
\begin{align*}
n=2L\log \sqrt{q},\\
u=\frac{\sqrt{q}^{\ell+1}-L}{L\sqrt{q}-\sqrt{q}^\ell}
\end{align*}
which is the same as solving the "polynomial-exponential" equation
\[
 \sqrt{q}^{\ell+1}-\sqrt{q}^{\ell}-(uL)-L=0
\]
which is hard to solve for $L$ in terms of $m$ or for $m$ in terms of $L$.
\end{remark}

\end{comment}

%************************************************
%******************* Section 5 ******************
%************************************************

\section{A Comparison between the three Curves}\label{sec4:5}
In this section we will use Mathematica \cite{Mathematica} to compare the performance of the three binary concatenated codes which are the Reed-Solomon code constructed from the projective line with $\theta=\sfrac{\pi}{2}$ (Example \ref{sec4:ex3.1}), the Hermitian code constructed from the Hermitian curve with $\theta=\sfrac{2\pi}{3}$ (Example \ref{sec4:ex3.2}), and the Suzuki code constructed from the Suzuki curve with $\theta=\sfrac{3\pi}{4}$ (Example \ref{sec4:ex3.3}). The two applications that will be considered in this section are the comparison between these codes as suitable kernels for channel polarization and as suitable codes for error-correction.

\subsection{A Comparison for Channel Polarization}\label{sec4:5.1}
As the rate of polarization is determined based on the exponent of the kernel, we compare numerically the three codes above as suitable kernels for channel polarization. Let $n$ be the binary block size, we draw a graph between the values of $n$ and the corresponding values of $E_2(G)$ for the three codes (see Figure \ref{comparsioneg}). 
\begin{figure}[htb!]
  \centering
  \includegraphics[height=30cm,width=10.0cm,keepaspectratio]{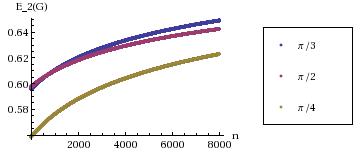}
  \caption[fixthetafreem]
  {The exponent $E_2(G)$ using $\theta=\sfrac{\pi}{3},\sfrac{\pi}{4},\sfrac{\pi}{2}$ and enumerating over the binary block size $n$.  We notice when $n=10000$, the three curves from top to bottom will correspond to $\theta=\sfrac{\pi}{3},\sfrac{\pi}{2},\sfrac{\pi}{4}$. For $n\geq 400,000$, the Suzuki code will be in the top, then the Hermitian code and finally the Reed-Solomon code}\label{comparsioneg}
\end{figure}
From the graph we notice that as $n$ gets larger, the maximum $E_2(G)$ is attained first by the Reed-Solomon codes for values of $n<1800$, after that as $n$ gets larger, by the Hermitian codes up to $n=399,212$ and after that by the Suzuki codes. This can be generalized to any $\theta=\sfrac{k\pi}{\ell}$, i.e., more geometry yields larger exponents as $n \to \infty$. We summarize this result in Figure \ref{fig:RSHS2}.

%%%%% Put here the figure of RS | H | S

\begin{figure}[htb!]
\centering

\begin{tikzpicture}
  [scale=1,auto=left, minimum size=5em]

   %%%%%%%%%%%%%% The following for the two Rectangles 
 
 \draw [<-] (0.5,4) --  (0.5,1); 
 \draw [->] (0.5,1) -- (8,1);

  \draw  (3.5,2) --  (3.5,3); 
  \draw  (5.5,2) --  (5.5,3);
  
 \node [align=center] (1) at (2.0,2.5) {Reed-Solomon};
 \node [] (2) at (4.5,2.5) {Hermitian};
 \node [align=center] (3) at (6.3,2.5) {Suzuki};

  \node [align=center] (1) at (3.5,0.7) {$1,800$}; 
   \node [align=center] (1) at (5.5,0.7) {$399,212$};

 \node [align=center] (1) at (8.3,1) {$n$}; 
 \node [align=center] (1) at (-0.5,3.75) { Maximum\\ $E_2(G)$}; 
  
\end{tikzpicture}
\caption{ The maximum exponent for a given $n$  is obtained with the indicated choice of curve.}\label{fig:RSHS2}
\end{figure}
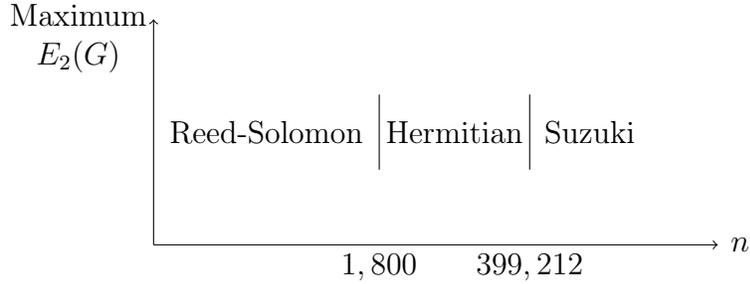

\subsection{A Comparison for Error-Correction}\label{sec4:5.2}
In this section we compare the performance of the three codes above as error-correcting codes with fixed binary rate. We draw a graph between the binary block size $n$ and $\sfrac{(k+d)}{n}$, where $n$, $k$, and $d$ are the parameters of the binary concatenated code with fixed binary rate.

For that, let $C$ be an algebraic geometry code with parameters $(L,K,D)$ over $\fg{2^m}$ constructed from the algebraic curve of genus $g$ with $L$ affine rational points. Then, we have 
\[
 L-g+1 \leq K+D < L+1 \Rightarrow \frac{L-g+1}{L} \leq \frac{K+D}{L}<\frac{L+1}{L}.
\]
Assume $\sfrac{(K+D)}{L}=\sfrac{(L-g+1)}{L}$ \cite[Corollary 4.1.14, p. 196]{TVN}. Then, the concatenated binary code has parameters $(mL,mK,d') (d'\geq D)$. Therefore,
\[
\frac{mK+d'}{mL}\geq \frac{mK+D}{mL}=\underbrace{\frac{(m-1)}{m}}_{\text{known}} \cdot \underbrace{\frac{K}{L}}_{\text{Rate}} + \underbrace{\frac{K+D}{mL}}_{\text{known}} =\frac{m-1}{m}\frac{K}{L}+\frac{L-g+1}{mL}.
\]
Therefore, we can draw the graph between $n$ and $\sfrac{(k+d)}{n}$ as in Figure \ref{comparsionag}.

\begin{figure}[htb!]
  \centering
  \includegraphics[height=30cm,width=10.0cm,keepaspectratio]{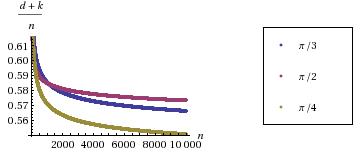}
  
  \caption[fixthetafreem]
  {The bound $\sfrac{(k+d)}{n}$ using $\theta=\sfrac{\pi}{3},\sfrac{\pi}{4},\sfrac{\pi}{2}$ when the binary rate $R_2:=\sfrac{1}{2}$ and enumerating over the binary block size $n$. We notice when $n=10000$, the three curves from top to bottom will be corresponding to $\theta=\sfrac{\pi}{2},\sfrac{\pi}{3},\sfrac{\pi}{4}$.}\label{comparsionag}
\end{figure}

Figure \ref{comparsionag} shows that as $n \to \infty$, $\sfrac{(k+d)}{n} \to R=\sfrac{1}{2}$ and for the first few values of $n$, the Reed-Solomon code is the closest code to the line $y=\sfrac{1}{2}$, then comes the Hermitian code to be the closest and as $n$ gets larger, the Suzuki code is the code that is closest to the line $y=\sfrac{1}{2}$, so we have once again that more geometry is preferable as the block size increases.

\bibliographystyle{amsplain}
\bibliography{mypaper,mypaper2,mypaper3,mypaper0}

\end{document}